\documentclass[12pt]{article}
\usepackage{amssymb, amsmath}
\usepackage{eepic}
\usepackage{epic}
\usepackage{array}
\usepackage{amscd}
\usepackage{amsthm}
\usepackage{graphicx}
\usepackage{tikz}
\usetikzlibrary{intersections,calc,arrows.meta}
\usepackage{a4wide}
\newtheorem{thm}{Theorem}[section]
\newtheorem{crl}[thm]{Corollary}

\newtheorem{lmm}[thm]{Lemma}

\newtheorem{prp}[thm]{Proposition}

\newtheorem{rem}{Remark}
\newcommand{\DIS}{\displaystyle}
\def\Z{\mathbb Z}
\def\R{\mathbb R}

\def\F{\mathcal F}

\def\III{I\hspace{-1.2pt}I\hspace{-1.2pt}I}
\def\II{I\hspace{-1.2pt}I}
\def\x{\mathfrak{x}}
\def\p{\mathfrak{p}}
\title{\textbf{The Volterra lattice, Abel's equation of the first kind, and the SIR epidemic models}}
\author{Atsushi \textsc{Nobe}\\[5pt]
\normalsize{Faculty of Political Science and Economics, Waseda University,}\\
\normalsize{1-6-1 Nishiwaseda, Shinjuku, Tokyo 169-8050, Japan}\\
\normalsize{e-mail: \texttt{nobe@waseda.jp}}
}
\date{}
\begin{document}
%

\maketitle

\begin{abstract}      
The Volterra lattice, when imposing non-zero constant boundary values, admits the structure of a completely integrable Hamiltonian system if the system size is sufficiently small.
Such a Volterra lattice can be regarded as an epidemic model known as the SIR model with vaccination, which extends the celebrated SIR model to account for vaccination.
Upon the introduction of an appropriate variable transformation, the SIR model with vaccination reduces to an Abel equation of the first kind, which corresponds to an exact differential equation.
The equipotential curve of the exact differential equation is the Lambert curve. 
Thus, the general solution to the initial value problem of the SIR model with vaccination, or the Volterra lattice with constant boundary values, is implicitly provided by using the Lambert W function.
\end{abstract}

\section{Introduction}\label{sec:intro}

The Volterra lattice is a simultaneous system of infinitely many first-order differential equations that pertain to the nodes on a one-dimensional infinite lattice \cite{KM75-1, KM75-2, Moser75}.
By imposing an appropriate boundary condition, the Volterra lattice reduces to a completely integrable Hamiltonian flow on a finite-dimensional Poisson manifold. 
Significant boundary conditions that contribute to complete integrability include the periodic boundary and the open-end boundary. 
When the Volterra lattice imposes either the periodic or the open-end boundary, it admits a bi-Hamiltonian structure on the Poisson manifold, thereby providing a sufficient number of conserved quantities \cite{Bogoyavlenskij91, Bogoyavlenskij08, Suris03} .

In this article, we consider another boundary condition that contributes to the complete integrability of the Volterra lattice: the boundary nodes are assigned constant values, which are not necessarily zero.
Although such a Volterra lattice has a Poisson structure, and hence, is a Hamiltonian flow on a finite-dimensional Poisson manifold, it does not admit a bi-Hamiltonian structure.
Thus, the Volterra lattice with constant boundary values is, in general, not a completely integrable Hamiltonian system.
Nevertheless, if the system size is sufficiently small, the Volterra lattice exhibits complete integrability because a sufficient number of conserved quantities immediately follow from the Hamiltonian.
If the system size is two, the Volterra lattice with constant boundary values reduces to an integrable epidemic model called the SIR model with vaccination \cite{Hethcote74}, which is an extension of the SIR model \cite{KM27} under the influence of vaccination and is abbreviated to the SIRv model.
If the system size is three, the Volterra lattice with constant boundary values also exhibits the complete integrability; however, its significance as an epidemic model is unclear.

Upon the introduction of an appropriate variable transformation, the SIRv model is transformed into an exact differential equation via a first-order nonlinear differential equation of degree three called Abel's equation of the first kind \cite{HLM14, MR16}.
The exact differential equation thus obtained possesses the potential arising from the symplectic structure of the Volterra lattice.
Moreover, the invariant curve of the SIRv model, or the equipotential curve of the Abel equation, is the Lambert curve.
Hence, the general solution to the SIRv model, or the Volterra lattice with constant boundary values, is implicitly provided in terms of the Lamber W function \cite{CGHJK96}.
In addition, an integrable discretization of the SIRv model, possessing exactly the same conserved quantity as the continuous model, is achieved through a geometric construction utilizing the Lambert curve \cite{Nobe23}.

This article is organized as follows.
In \$\ref{sec:VL}, we introduce the Volterra lattice and briefly review its Poisson structure.
Then we show that the Volterra lattice with constant boundary values is a completely integrable Hamiltonian flow on the Poisson manifold if the system size is either two or three.
We moreover show that the two-dimensional Volterra lattice can be regarded as an epidemic model called the SIRv model.
In \$\ref{sec:AbelSIRv}, we investigate Abel's equation of the first kind and show that it reduces to an exact differential equation if some conditions are satisfied.
We also show that the Abel equation can be related with the SIRv model.
The equipotential curve of the exact differential equation is the Lambert curve, thereby the general solution to the initial value problems of the SIRv model, or the two-dimensional Volterra lattice with constant boundary values, is implicitly provided in terms of the Lambert W function.
We devote to concluding remarks in \$\ref{sec:concl}.

\section{The Volterra lattice}\label{sec:VL}
Let $p$ a natural number.
The following simultaneous system of inifinitely many first-order differential equations that pertain to the nodes on a one-dimensional infinite lattice is called the Volterra lattice or the Lotka-Volterra system \cite{KM75-1, KM75-2, Moser75, Bogoyavlenskij91,Bogoyavlenskij08}
\begin{align}
\dot a_i
=
a_i
\left(
\sum_{j=1}^pa_{i+j}
-
\sum_{j=1}^pa_{i-j}
\right)
\qquad
(i\in\Z),
\label{eq:LVp}
\end{align}
where $a_i=a_i(t)$ is a differentiable function in $t$ assigned to the $i$-th node, and ``$\ \dot{}\ $" denotes the derivative with respect to $t$.
Throughout this article, we assume $p=1$, which achieves the simplification of \eqref{eq:LVp}:
\begin{align}
\dot a_i
=
a_i
\left(
a_{i+1}
-
a_{i-1}
\right)
\qquad
(i\in\Z).
\label{eq:LV}
\end{align}

\subsection{Boundary conditions}\label{subsec:BC}

Let us consider the differentiable $L$-dimensional manifold
\begin{align*}
V=\R^{L}(a_1,a_2,\ldots,a_{L}),
\end{align*}
where $(a_1,a_2,\ldots,a_{L})$ stands for the local, and hence, the global coordinates.
The Volterra lattices equipped with the following two boundary conditions are known as completely integrable Hamiltonian systems on the phase space $V$ \cite{Bogoyavlenskij91,Bogoyavlenskij08,Suris03}.
\begin{enumerate}
\item[(I)] Periodic boundary ($a_i=a_{L+i}$).

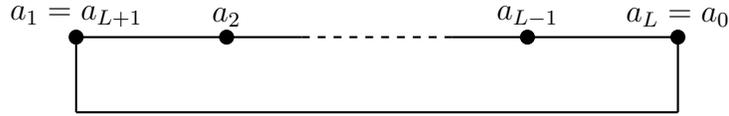
\begin{align*}
\begin{tikzpicture}
\draw[thick](0,0)--(3,0);
\draw[dashed, thick](3,0)--(5,0);
\draw[thick](5,0)--(8,0);
\draw[thick](0,0)--(0,-1);
\draw[thick](0,-1)--(8,-1);
\draw[thick](8,-1)--(8,0);
\fill[black](0,0)circle(0.1)node[above]{$a_1=a_{L+1}$};
\fill[black](2,0)circle(0.1)node[above]{$a_2$};
\fill[black](6,0)circle(0.1)node[above]{$a_{L-1}$};
\fill[black](8,0)circle(0.1)node[above]{$a_L=a_0$};
\end{tikzpicture}
\end{align*}

\item[(\II)] Open-end boundary ($a_0=a_{L}=0$).

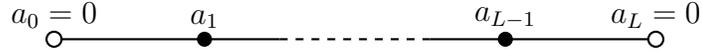
\begin{align*}
\begin{tikzpicture}
\draw[thick](-1.91,0)--(1,0);
\draw[dashed, thick](1,0)--(3,0);
\draw[thick](3,0)--(5.91,0);
\fill[black](0,0)circle(0.1)node[above]{$a_1$};
\fill[black](4,0)circle(0.1)node[above]{$a_{L-1}$};
\draw[thick](6,0)circle(0.1)node[above]{$a_{L}=0$};
\draw[thick](-2,0)circle(0.1)node[above]{$a_{0}=0$};
\end{tikzpicture}
\end{align*}

The exterior nodes $a_{-1},a_{-2},\ldots$ and $a_{L+1},a_{L+2},\ldots$ are not uniquely determined, even when the values of the inner nodes $a_0,a_1,\ldots,a_{L}$ are assigned.
Actually, since
\begin{align*}
0
&=
\dot a_0
=
a_0
\left(
a_{1}
-
a_{-1}
\right)
=
0\left(
a_{1}
-
a_{-1}
\right),
\\
0
&=
\dot a_{L}
=
a_{L}
\left(
a_{L+1}
-
a_{L-1}
\right)
=
0
\left(
a_{L+1}
-
a_{L-1}
\right)
\end{align*}
holds, the nodes $a_{-1}$ and $a_{L+1}$ are arbitrary, thereby we inductively arrive at the claim.
Since $a_0=a_L=0$ and the exterior nodes are arbitrary, the Volterra lattice with the open-end boundary is a Hamiltonian system on the phase space $V=\R^{L-1}$, making it a subsystem of the one with the periodic boundary (I) on $V=\R^L$. 
\end{enumerate}

In this article, we consider the third boundary condition with which the Volterra lattice \eqref{eq:LV} is completely integrable Hamiltonian system on $V=\R^M$ for certain system sizes $M$:
\begin{itemize}
\item[(\III)] Constant boundary ($a_0=\alpha$, $a_{M+1}=\beta$, $\alpha,\beta\in\R\backslash\{0\}$).

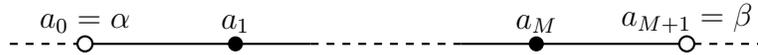
\begin{align*}
\begin{tikzpicture}
\draw[thick](-1.91,0)--(1,0);
\draw[dashed, thick](1,0)--(3,0);
\draw[thick](3,0)--(5.91,0);
\draw[dashed, thick](-3,0)--(-2.09,0);
\draw[dashed, thick](6.09,0)--(7,0);
\fill[black](0,0)circle(0.1)node[above]{$a_1$};
\fill[black](4,0)circle(0.1)node[above]{$a_{M}$};
\draw[thick](6,0)circle(0.1)node[above]{$a_{M+1}=\beta$};
\draw[thick](-2,0)circle(0.1)node[above]{$a_{0}=\alpha$};
\end{tikzpicture}
\end{align*}

Contrary to (\II), the exterior nodes $a_{-1},a_{-2},\ldots$ and $a_{M+2},a_{M+3},\ldots$ are uniquely determined as rational functions of the inner nodes $a_0,a_1,\ldots,a_{M+1}$ via \eqref{eq:LV}:
\begin{align*}
a_{i-1}
&=
a_{i+1}-\frac{\dot a_i}{a_i}
\qquad(i<0),
\\
a_{i+1}
&=
a_{i-1}+\frac{\dot a_i}{a_i}
\qquad(i>M+1).
\end{align*}
In particular, the nearest exterior nodes are reflective with respect to the boundary nodes $a_0$ and $a_{M+1}$,
\begin{align*}
a_{-1}=a_1
\quad
\mbox{and}
\quad
a_{M+2}=a_{M},
\end{align*}
since $a_0\neq0$, $a_{M+1}\neq0$ and $\dot a_0=\dot a_{M+1}=0$.
Hence, the system with the constant boundary is not a subsystem of the one with the periodic boundary (I).

\end{itemize}

Numerical examination implies that the Volterra lattice \eqref{eq:LV} with the constant boundary (\III) is, in general, not completely integrable.
Nevertheless, the Volterra lattice with the constant boundary (\III) exhibits the complete integrability if the system size $M$ is sufficiently small, since it has a Poisson structure on $V=\R^M$ for any $M$.

\subsection{Poisson structures}\label{subsec:PS}
Before reviewing the Poisson structure, we introduce the two-field form of the Volterra lattice.
Let $a_{2k}=x_k$ and $a_{2k-1}=y_k$.
Then the Volterra lattice \eqref{eq:LV} reduces to the system of first-order ODEs called the two-field form
\begin{align}
\dot x_k
&=
x_k
\left(
y_{k+1}
-
y_k
\right),
\label{eq:VL2ff1}\\
\dot y_k
&=
y_k
\left(
x_k
-
x_{k-1}
\right)
\qquad
\label{eq:VL2ff2}
\end{align}
for $k\in\Z$.

First remark the bi-Hamiltonian structure of the Volterra lattice (\ref{eq:VL2ff1}--\ref{eq:VL2ff2}) on the $2N$-dimensional phase space
\begin{align*}
V=\R^{2N}(x_1,\ldots,x_N,y_1,\ldots,y_N)
\end{align*}
with the periodic boundary (I).
We adopt the following convention: the Poisson brackets are defined by writing down all non-vanishing brackets between the coordinate functions.
We say that two Poisson brackets on $V$ are compatible if their arbitrary linear combination is also a Poisson bracket on $V$.

\begin{prp}[\cite{Suris03}]\label{prop:H0H1}\normalfont
Suppose the periodic boundary condition (I).
The relations
\begin{align*}
\left\{x_k,y_k\right\}_2&=x_ky_k,
\\
\left\{x_k,y_{k+1}\right\}_2&=-x_ky_{k+1}
\end{align*}
define a quadratic Poisson bracket $\{\cdot,\cdot\}_2$ on $V$.
The flow (\ref{eq:VL2ff1}--\ref{eq:VL2ff2}) is a Hamiltonian system on $(V,\{\cdot,\cdot\}_2)$ with the Hamiltonian function
\begin{align*}
H_1(x,y)
:=
\sum_{k=1}^{N}(x_k+y_k).
\end{align*}

Also the relations
\begin{align*}
\left\{x_k,y_k\right\}_3&=x_ky_k(x_k+y_k),
\\
\left\{x_k,y_{k+1}\right\}_3&=-x_ky_{k+1}(x_k+y_{k+1}),
\\
\left\{x_k,x_{k+1}\right\}_3&=-x_ky_{k+1}x_{k+1},
\\
\left\{y_k,y_{k+1}\right\}_3&=-y_kx_ky_{k+1}
\end{align*}
define a cubic Poisson bracket $\{\cdot,\cdot\}_3$ on $V$ compatible with $\{\cdot,\cdot\}_2$.
The flow (\ref{eq:VL2ff1}--\ref{eq:VL2ff2}) is a Hamiltonian system on $(V,\{\cdot,\cdot\}_3)$ with the Hamiltonian function
\begin{align*}
H_0(x,y)
&:=
\frac{1}{2}\sum_{k=1}^{N}\left(\log x_k+\log y_k\right).
\end{align*}
These Poisson structures are also valid for the open-end boundary (\II).
\qed
\end{prp}

Let $\F(V)$ be the set of smooth real-valued functions on the $M$-dimensional manifold $V$.
For $F,G\in\F(V)$ and a Poisson bracket $\{\cdot,\cdot\}$ on $V$, we have
\begin{align*}
\{F,G\}
&=
\sum_{i,j=1}^{M}A_{ij}\frac{\partial F}{\partial a_i}\frac{\partial G}{\partial a_j},
\end{align*} 
where $A_{ij}=\{a_i,a_j\}$ forms the skew-symmetric $M\times M$ matrix as a coordinate representation of the Poisson tensor.

The assertion in Proposition \ref{prop:H0H1} can be confirmed as follows.
For the cubic Poisson bracket $\{\ ,\ \}_3$, we have
\begin{align*}
\left\{H_0,x_k\right\}_3
&=
\frac{1}{2}\left(
\left\{\log x_{k-1},x_k\right\}_3+\left\{\log x_{k+1},x_k\right\}_3+\left\{\log y_k,x_k\right\}_3+\left\{\log y_{k+1},x_k\right\}_3
\right)
\\
&=
\frac{1}{2}\left(
-\frac{x_{k-1}y_kx_k}{x_{k-1}}+\frac{x_ky_{k+1}x_{k+1}}{x_{k+1}}-\frac{x_ky_k(x_k+y_k)}{y_k}+\frac{x_ky_{k+1}(x_k+y_{k+1})}{y_{k+1}}
\right)
\\
&=
x_k\left(y_{k+1}-y_k\right).
\end{align*}
For the quadratic bracket $\{\ ,\ \}_2$, we also have
\begin{align*}
\left\{H_1,x_k\right\}_2
&=
\left\{y_k,x_k\right\}_2+\left\{y_{k+1},x_k\right\}_2
=
x_k\left(y_{k+1}-y_k\right).
\end{align*}
We similarly obtain 
\begin{align*}
\left\{H_0,y_k\right\}_3
&=
\left\{H_1,y_k\right\}_2
=
y_k(x_k-x_{k-1}).
\end{align*}
Thus (\ref{eq:VL2ff1}--\ref{eq:VL2ff2}) defines two compatible Hamiltonian flows on $V=\R^{2N}$:
\begin{align*}
\dot x_k&=\left\{H_1,x_k\right\}_2=\left\{H_0,x_k\right\}_3,
\\
\dot y_k&=\left\{H_1,y_k\right\}_2=\left\{H_0,y_k\right\}_3.
\end{align*}

Meanwhile, suppose the Volterra lattice (\ref{eq:VL2ff1}--\ref{eq:VL2ff2}) to have the constant boundary (\III).
If we consider the phase space
\begin{align*}
V=\R^{2N}(x_1,\ldots,x_N,y_1,\ldots,y_N)
\end{align*}
of $2N$-dimension then (\ref{eq:VL2ff1}--\ref{eq:VL2ff2}) reduces to
\begin{align}
\dot y_1&=y_1(x_1-\alpha),
\label{eq:VL2ffeven11}\\
\dot x_k
&=
x_k
\left(
y_{k+1}
-
y_k
\right)
&&
\mbox{for $k=1,2,\ldots,N-1$},&&
\label{eq:VL2ffeven12}\\
\dot y_k
&=
y_k
\left(
x_k
-
x_{k-1}
\right)
&&
\mbox{for $k=2,3,\ldots,N$},&&
\label{eq:VL2ffeven21}\\
\dot x_N
&=
x_N(\beta-y_N).
\label{eq:VL2ffeven22}
\end{align}
Whereas, if the phase space is of $(2N-1)$-dimension,
 \begin{align*}
V=\R^{2N-1}(x_1,\ldots,x_{N-1},y_1,\ldots,y_{N}),
\end{align*}
then
\begin{align}
\dot y_1&=y_1(x_1-\alpha),
\label{eq:VL2ffodd11}\\
\dot x_k
&=
x_k
\left(
y_{k+1}
-
y_k
\right)
&&
\mbox{for $k=1,2,\ldots,N-1$},&&
\label{eq:VL2ffodd12}\\
\dot y_k
&=
y_k
\left(
x_k
-
x_{k-1}
\right)
&&
\mbox{for $k=2,3,\ldots,N-1$},&&
\label{eq:VL2ffodd21}\\
\dot y_N
&=
y_N(\beta-x_{N-1}).
\label{eq:VL2ffodd22}
\end{align}

Let us consider the following function $H_{01}(x,y)$ in $x_k$ and $y_k$
\begin{align}
H_{01}(x,y)
&:=
\sum_{k=1}^{N}(x_k+y_k)
-
\alpha\sum_{k=1}^{N}\log x_k
-
\beta\sum_{k=1}^{N}\log y_k.
\label{eq:H1xy}
\end{align}
If $M=2N$ then the boundaries are $x_0=\alpha$ and $y_{N+1}=\beta$, and the derivative of $H_{01}$ with respect to $t$ always vanishes
\begin{align*}
\dot H_{01}(x,y)
=&
\sum_{k=1}^N\left(x_ky_{k+1}-x_{k-1}y_k\right)
-
\alpha\sum_{k=1}^{N}(y_{k+1}-y_k)
-
\beta\sum_{k=1}^{N}(x_k-x_{k-1})
\\
=&
\beta x_N-\alpha y_1-\alpha(\beta-y_1)-\beta(x_N-\alpha)
=
0.
\end{align*}
Therefore, $H_{01}$ is a conserved quantity of the system (\ref{eq:VL2ffeven11}--\ref{eq:VL2ffeven22}) for arbitrary $\alpha,\beta$.

If $M=2N-1$ then the boundaries are $x_0=\alpha$ and $x_N=\beta$, but $H_{01}$ is, in general, not a conserved quantity; the derivative of $H_{01}$ with respect to $t$ does not vanish:
\begin{align*}
\dot H_{01}(x,y)
=&
x_{N-1}y_N-\alpha y_1+y_N(\beta -x_{N-1})-\alpha (y_N-y_1)-\beta(\beta -\alpha )
\\
=&
(y_N-\beta)(\beta -\alpha ).
\end{align*}
However, this computation suggests that if $\alpha =\beta$ then $H_{01}$ is still a conserved quantity.
Moreover, it implies that $H_{01}$ is divided into two conserved quantities $G_1$ and $G_2$ as $H_{01}=G_1-\alpha G_2$ if $\alpha =\beta$, where
\begin{align*}
G_1(x,y)
&:=
\sum_{k=1}^{N}(x_k+y_k)
-
\alpha\sum_{k=1}^{N}\log x_k,
\\
G_2(x,y)
&:=
\sum_{k=1}^{N}\log y_k.
\end{align*}

Now, investigate the Poisson structure of the Volterra lattice (\ref{eq:VL2ff1}--\ref{eq:VL2ff2}) equipped with the constant boundary (\III), or (\ref{eq:VL2ffeven11}--\ref{eq:VL2ffeven22}) and  (\ref{eq:VL2ffodd11}--\ref{eq:VL2ffodd22}).
We easily see that the Hamiltonians $H_0$ and $H_1$ in Proposition \ref{prop:H0H1} are no longer conserved when the constant boundary condition (\III) is imposed.
However, we know that the function $H_{01}$ (see \eqref{eq:H1xy}) is still a conserved quantity of (\ref{eq:VL2ffeven11}--\ref{eq:VL2ffeven22}) for arbitrary $\alpha,\beta$, and is of (\ref{eq:VL2ffodd11}--\ref{eq:VL2ffodd22}) when $\alpha=\beta$.
We then obtain the following proposition concerning the Poisson structure of the Volterra lattice with the constant boundary (\III).

\begin{prp}\label{prop:H01}\normalfont
Suppose the constant boundary condition (\III).
If $V=\R^{2N}$, the flow  (\ref{eq:VL2ffeven11}--\ref{eq:VL2ffeven22}) is the Hamiltonian system on the Poisson manifold $(V,\{\cdot,\cdot\}_2)$ with the Hamiltonian function $H_{01}$ for any $\alpha,\beta$.
Similarly, if $V=\R^{2N-1}$ and $\alpha=\beta$, the flow (\ref{eq:VL2ffodd11}--\ref{eq:VL2ffodd22}) is also the Hamiltonian system on $(V,\{\cdot,\cdot\}_2)$ with $H_{01}$.
\end{prp}

\begin{proof}
We have
\begin{align*}
\left\{H_{01},\log x_k\right\}_2
=&-\frac{x_ky_k}{x_k}+\frac{x_ky_{k+1}}{x_k}+\beta\frac{x_ky_k}{x_ky_k}-\beta\frac{x_ky_{k+1}}{x_ky_{k+1}}
=y_{k+1}-y_k,
\\
\left\{H_{01},\log y_k\right\}_2
=&
\frac{x_ky_k}{y_k}-\frac{x_{k-1}y_{k}}{y_k}-\alpha \frac{x_ky_k}{x_ky_k}+\alpha\frac{x_{k-1}y_{k}}{x_{k-1}y_{k}}
=x_{k}-x_{k-1}
\end{align*}
for non-boundary nodes $x_k$ and $y_k$.
If $V=\R^{2N}$, these are also valid for the boundary nodes $y_1$ and $x_N$ for any boundary values $\alpha,\beta$, since we have
\begin{align*}
\left\{H_{01},\log y_1\right\}_2
=&
\frac{x_1y_1}{y_1}-\alpha \frac{x_1y_1}{x_1y_1}
=x_1-\alpha,
\\
\left\{H_{01},\log x_N\right\}_2
=&
-\frac{x_Ny_N}{x_N}+\beta\frac{x_Ny_N}{x_Ny_N}
=\beta-y_N.
\end{align*}
Whereas, if $V=\R^{2N-1}$ and $\alpha=\beta$, at the boundary, we also have $\left\{H_{01},\log y_1\right\}_2=x_1-\alpha$ and
\begin{align*}
\left\{H_{01},\log y_N\right\}_2
=&
-\frac{x_{N-1}y_N}{y_N}+\alpha\frac{x_{N-1}y_N}{x_{N-1}y_N}
=\alpha-x_{N-1}=\beta-x_{N-1}.
\end{align*}

Hence the Hamiltonian flow
\begin{align*}
\frac{d}{dt}\log x_k
&=
\left\{H_{01},\log x_k\right\}_2
=
y_{k+1}-y_k,
\\
\frac{d}{dt}\log y_k
&=
\left\{H_{01},\log y_k\right\}_2
=
x_{k}-x_{k-1}
\end{align*}
is equivalent to the Volterra lattice (\ref{eq:VL2ff1}--\ref{eq:VL2ff2}) for both $V=\R^{2N}$ with arbitrary $\alpha,\beta$ and $V=\R^{2N-1}$ with $\alpha=\beta$.
\end{proof}

\begin{rem}
In the limit as $\alpha,\beta\to0$, the constant boundary condition (\III) reduces to the open-end one (\II), and $H_{01}$ consistently approaches $H_1$, the Hamiltonian of the Volterra lattice with the open-end boundary (\II) with respect to the Poisson bracket $\{\ ,\ \}_2$.
\end{rem}

For the Poisson bracket $\{\ ,\ \}_2$ on the $2N$-dimensional phase space
\begin{align*}
V=\R^{2N}(x_1,\ldots,x_N,y_1,\ldots,y_N),
\end{align*}
we have the skew-symmetric $2N\times 2N$ matrix $A=(A_{ij})$ of the Poisson tensor that possess the following non-zero entries
\begin{align*}
A_{2k-1,2k}&=-x_ky_k&&\mbox{and}&A_{2k,2k-1}&=x_ky_k&&\mbox{for $k=1,2,\ldots,N$},
\\
A_{2k,2k+1}&=-x_ky_{k+1}&&\mbox{and}&A_{2k+1,2k}&=x_ky_{k+1}&&\mbox{for $k=1,2,\ldots,N-1$}.
\end{align*}
The matrix $A$ is non-degenerate, thereby we also have a symplectic structure on $V$.

In order to capture the symplectic structure, we introduce new coordinate variables
\begin{align*}
\x_k&=\log y_1+\log y_2+\cdots+\log y_k,
\\
\p_k&=\log x_k
\end{align*}
for $k=1,2,\ldots,N$.
Then we have
\begin{align*}
\{\p_k,\x_k\}_2
&=
\left\{\log x_k,\log y_k\right\}_2
=
1
\end{align*}
for $k=1,2,\ldots,N$ and
\begin{align*}
\{\p_k,\x_{k+1}\}_2
&=
\left\{\log x_k,\log y_k\right\}_2+\left\{\log x_k,\log y_{k+1}\right\}_2
=
0
\end{align*}
for $k=1,2,\ldots,N-1$.
Hence $\x_1,\ldots,\x_N,\p_1,\ldots,\p_N$ form canonically conjugate coordinates on the symplectic manifold $(V=\R^{2N},\Omega)$ possessing the  symplectic form 
\begin{align}
\Omega
&=
\sum_{k=1}^Nd\p_k\wedge d\x_k
=
\sum_{k=1}^N\sum_{\ell=1}^k\frac{dx_k\wedge dy_\ell}{x_ky_\ell}.
\label{eq:SF}
\end{align}
The symplectic manifold $(V,\Omega)$ is called the canonical phase space of the Hamiltonian flow (\ref{eq:VL2ffeven11}--\ref{eq:VL2ffeven22}).

In these canonical coordinates, the hamiltonian $H_{01}$ is represented by
\begin{align*}
H_{01}(\p,\x)
&=
\sum_{k=1}^Ne^{\p_k}
+
\sum_{k=1}^Ne^{\x_k-\x_{k-1}}
-
\alpha\sum_{k=1}^N\p_k
-
\beta\x_N,
\end{align*}
where we assume $\x_0=0$.
Hence we have
\begin{align*}
\frac{\partial H_{01}}{\partial\p_k}
&=
e^{\p_k}-\alpha
=
x_k-\alpha,
\\
\frac{\partial H_{01}}{\partial\x_k}
&=
e^{\x_k-\x_{k-1}}-e^{\x_{k+1}-\x_k}
=
y_k-y_{k+1},
\\
\dot\x_k
&=
\sum_{\ell=1}^k\frac{\dot y_\ell}{y_\ell}
=
\sum_{\ell=1}^k(x_\ell-x_{\ell-1})
=
x_k-\alpha,
\\
\dot\p_k
&=
\frac{\dot x_k}{x_k}
=
y_{k+1}-y_k.
\end{align*}
Therefore, the Volterra lattice (\ref{eq:VL2ffeven11}--\ref{eq:VL2ffeven22}) is canonically represented by the Hamilton equations
\begin{align}
\dot\x_k
&=
\frac{\partial H_{01}}{\partial\p_k},
\label{eq:canonicalHamiltonianFlow1}\\
\dot\p_k
&=
-\frac{\partial H_{01}}{\partial\x_k}
\label{eq:canonicalHamiltonianFlow2}
\end{align}
for $k=1,2,\ldots,N$.

Thus, we obtain:
\begin{prp}\label{prop:Symplectic}\normalfont
Suppose the constant boundary condition (\III).
The flow (\ref{eq:VL2ffeven11}--\ref{eq:VL2ffeven22}) is a Hamiltonian system on the symplectic manifold $(V=\R^{2N},\Omega)$ with the Hamiltonian function $H_{01}$ with respect to the symplectic form $\Omega$ given by \eqref{eq:SF}.
\qed
\end{prp}

For the Poisson bracket $\{\ ,\ \}_2$ on the $(2N-1)$-dimensional phase space
\begin{align*}
V=\R^{2N-1}(x_1,\ldots,x_{N-1},y_1,\ldots,y_N),
\end{align*}
the skew-symmetric $(2N-1)\times (2N-1)$ matrix of the Poisson tensor is degenerate, thereby it does not define a symplectic structure on $V$.

It should be remarked that the Hamiltonian $H_{01}$ and the functions $G_i\in\F(V)$ ($i=1,2$) are in involution with respect to the Poisson bracket $\{\ ,\ \}_2$ on $V=\R^{2N-1}$,
\begin{align*}
\left\{
H_{01},G_1
\right\}_2
=
\left\{
H_{01},G_2
\right\}_2
=0.
\end{align*}
Indeed, we have $H_{01}=G_1-\alpha G_2$ and
\begin{align*}
\left\{
G_1,G_2
\right\}_2
=&
\sum_{k=1}^{N}
\left\{
x_k-\alpha\log x_k,
\log y_k
\right\}_2
+
\sum_{k=1}^{N-1}
\left\{
x_k-\alpha\log x_k,
\log y_{k+1}
\right\}_2
\\
=&
\sum_{k=1}^{N}(x_k-\alpha)
-
\sum_{k=1}^{N-1}(x_k-\alpha)
\\
=&
x_N-\alpha 
=
0,
\end{align*}
where we use the assumption $x_N=\beta=\alpha$.
Thus, if $\alpha=\beta$, there exist at least two functionally independent conserved quantities for the Hamiltonian flow (\ref{eq:VL2ffodd11}--\ref{eq:VL2ffodd22}) on the Poisson manifold $(V=\R^{2N-1},\{\ ,\ \}_2)$.
 
\subsection{Complete integrability}\label{subsec:CI}

By virtue of the above observation, we easily find two completely integrable Hamiltonian systems equipped with the constant boundary (\III).

The first one is the case where $M=2N$ and $N=1$.
The Hamiltonian $H_{01}$ achieves a sufficient number of conserved quantities since the phase space $V$ is of two-dimension.
The Hamiltonian flow (\ref{eq:VL2ffeven11}--\ref{eq:VL2ffeven22}) then reduces to
\begin{align}
\dot x_1
&=
x_1
\left(
\beta-y_1
\right),
\label{eq:VL2ffeven2dim1}\\
\dot y_1
&=
y_1
\left(
x_1-\alpha
\right),
\qquad
\label{eq:VL2ffeven2dim2}
\end{align}
and the Hamiltonian
\begin{align*}
H_{01}(x_1,y_1)
&=
x_1+y_1-\alpha \log x_1-\beta\log y_1
\end{align*}
to the conserved quantity.

If $\alpha >0$ and as $\beta \to0$, (\ref{eq:VL2ffeven2dim1}--\ref{eq:VL2ffeven2dim2}) is nothing but the SIR epidemic model \cite{KM27}, which is  known as an integrable dynamical system crucial for mathematical analysis on the spread of infectious diseases.
Moreover, if $\alpha >0$ and $\beta <0$, (\ref{eq:VL2ffeven2dim1}--\ref{eq:VL2ffeven2dim2}) is an integrable extension of the SIR model under the influence of vaccination, called the SIRv model \cite{Hethcote74}.
In the next section, we will investigate (\ref{eq:VL2ffeven2dim1}--\ref{eq:VL2ffeven2dim2}) precisely, and reveal the relation with an exact differential equation via Abel's equation of the first kind.
We moreover provide the general solution to the initial value problem of (\ref{eq:VL2ffeven2dim1}--\ref{eq:VL2ffeven2dim2}) in terms of the Lambert W function.

As mentioned earlier, the symplectic structure of (\ref{eq:VL2ffeven2dim1}--\ref{eq:VL2ffeven2dim2}) is explicitly given as
\begin{align}
\Omega
&=
d\p\wedge d\x
=
\frac{dx_1\wedge dy_1}{x_1y_1},
\label{eq:SF2dim}
\end{align}
where the canonically conjugate coordinates $\x$ and $\p$ are given by
\begin{align*}
\x=\log y_1
\qquad
\mbox{and}
\qquad
\p=\log x_1.
\end{align*}
The Hamiltonian in the canonically conjugate coordinates,
\begin{align*}
H_{01}(\p,\x)
&=
e^\p+e^\x-\alpha \p-\beta \x,
\end{align*}
solves Hamilton's canonical equations (\ref{eq:canonicalHamiltonianFlow1}--\ref{eq:canonicalHamiltonianFlow2}) of motion with $N=1$.

The case where $M=2N-1$ and $N=2$ includes the second completely integrable system.
In this three-dimensional system, we moreover assume $\alpha=\beta$.
Then the Hamiltonian flow (\ref{eq:VL2ffodd11}--\ref{eq:VL2ffodd22}) reduces to the system
\begin{align}
\dot y_1
&=
y_1
\left(
x_1
-
\alpha
\right),
\label{eq:VL2ffodd3dim1}\\
\dot x_1
&=
x_1
\left(
y_2
-
y_1
\right),
\label{eq:VL2ffodd3dim2}\\
\dot y_2
&=
y_2
\left(
\alpha
-
x_1
\right),
\label{eq:VL2ffodd3dim3}
\end{align}
which possesses two conserved quantities:
\begin{align*}
G_1(x_1,y_1,y_2) &= y_1 + x_1 + y_2 - \alpha\log x_1,
\\
G_2(x_1,y_1,y_2) &= \log y_1 + \log y_2.
\end{align*}
However, having these conserved quantities is sufficient for (\ref{eq:VL2ffodd3dim1}--\ref{eq:VL2ffodd3dim3}) to exhibit complete integrability, since the phase space $V$ is of three-dimension.

We can eliminate $y_2$ from (\ref{eq:VL2ffodd3dim1}--\ref{eq:VL2ffodd3dim3})  by employing the conserved quantity $G_2$.
Let $G_2$ be a constant $\log\ell$ ($\ell>0$).
Then we have $y_1y_2=\ell$.
By replacing $y_2$ with $\ell/y_1$, (\ref{eq:VL2ffodd3dim1}--\ref{eq:VL2ffodd3dim2}) and $G_1$ respectively reduce to
\begin{align*}
\dot y_1
&=
y_1(x_1-\alpha),
\\
\dot x_1
&=
x_1\left(\frac{\ell}{y_1}-y_1\right)
\end{align*}
and
\begin{align*}
G_1(x_1,y_1)
&=
y_1+x_1+\frac{\ell}{y_1}-\alpha\log x_1.
\end{align*}
This is a two-dimensional completely integrable system possessing the conserved quantity $G_1$, but it differs slightly from (\ref{eq:VL2ffeven2dim1}--\ref{eq:VL2ffeven2dim2}), or the SIRv model. 
In the limit as $\ell\to 0$, it approaches the SIR model. However, its significance as an epidemic model has remained unclear.

\section{Abel's equation of the first kind and the SIRv model}\label{sec:AbelSIRv}
In the previous section, we observed that the SIRv model is a completely integrable Hamiltonian system on the symplectic manifold $(V=\R^2,\Omega)$, where $\Omega$ is the symplectic form \eqref{eq:SF2dim}.
Meanwhile, in this section, we investigate the SIRv model from another perspective on integrability through a first-order nonlinear differential equation of degree three known as Abel's equation of the first kind \cite{MR16}. Abel's equation of the first kind is a generalization of the Riccati equation, and similar to the Riccati equation, it admits the completely integrability under certain conditions.

\subsection{Abel's equation of the first kind}\label{subsec:Abel}
Let us consider the following first-order nonlinear ODE of degree three pertaining to $\phi=\phi(x)$
\begin{align}
\frac{d\phi}{dx}
&=
f_0
+
f_1\phi
+
f_2\phi^2
+
f_3\phi^3,
\label{eq:AbelODE1k}
\end{align}
where  $f_0$, $f_1$, $f_2$ and $f_3$ are meromorphic functions in $x$.
The equation \eqref{eq:AbelODE1k} is called Abel's equation of the first kind if $f_3\not\equiv0$, whereas the Riccati equation if $f_3\equiv0$ and $f_2\not\equiv0$.

Hereafter we assume $f_3\not\equiv0$, which makes \eqref{eq:AbelODE1k} Abel's equation of the first kind.
We have the following lemma \cite{MR16}.

\begin{lmm}\label{lmm:Abelexact}
If $f_0$, $f_1$, $f_2$ and $f_3$ satisfy
\begin{align}
f_0\equiv&0,
\label{eq:f0ODE}\\
\frac{d}{dx}\log\frac{f_2}{f_3}=&f_1
\label{eq:f1f2f3ODE}
\end{align}
then \eqref{eq:AbelODE1k} reduces to an exact differential equation.
\end{lmm}

\begin{proof}
First introduce a new dependent variable $\varphi=\varphi(x)$ such that
\begin{align*}
\phi\varphi=\exp\left(\int f_1dx\right)=:F(x).
\end{align*}
If $f_0\equiv0$, \eqref{eq:AbelODE1k} reduces to the following ODE concerning $\varphi$:
\begin{align}
\varphi\frac{d\varphi}{dx}+f_2F\varphi+f_3F^2=0.
\label{eq:ExactAbel}
\end{align}

Next consider a function $\varpi$ in $x$ and $\varphi$:
\begin{align*}
\varpi(x,\varphi)
&=
e^{-(1+k\varphi)-kG},
\end{align*}
where
\begin{align*}
G(x)
:=
\int f_2Fdx
=
\int f_2\exp\left(\int f_1dx\right)dx
\end{align*}
and $k$ is the integration constant of \eqref{eq:f1f2f3ODE}:
\begin{align*}
k
&=
\frac{f_2}{f_3F}.
\end{align*}
By multiplying the integrating factor $\varpi$, \eqref{eq:ExactAbel} reduces to the exact differential equation
\begin{align}
d\Psi(x,\varphi)
&=
0
\label{eq:Exact}
\end{align}
possessing the potential
\begin{align*}
\Psi(x,\varphi)
&=
(1+k\varphi)\varpi(x,\varphi).
\end{align*}
Indeed, we compute
\begin{align*}
d\Psi(x,\varphi)
&=
-k^2\varpi\left[\left(f_2F\varphi+f_3F^2\right)dx+\varphi d\varphi\right],
\end{align*}
which asserts the equivalence of \eqref{eq:ExactAbel} and \eqref{eq:Exact}.
\end{proof}

The solution to the exact differential equation \eqref{eq:Exact} is provided by using the Lambert W function $W(z)$. 
Through the relationship $w=W(z)$, this function parametrizes the following non-algebraic curve on the $(z,w)$-plane, known as the Lambert curve \cite{CGHJK96}:
\begin{align*}
\left(we^w-z=0\right).
\end{align*}

The Lambert W function $W(z)$ is single-real-valued on $[0,\infty)$, while it is double-real-valued on $[-e^{-1},0)$ and is not defined on $(-\infty,-e^{-1})$ as a real function.
We denote the brach such that $-1\leq W(z)<0$ by $W_0(z)$, and the one such that $W(z)<-1$ by $W_{-1}(z)$.
Hence the function
\begin{align*}
W(z)
&=
\begin{cases}
W_0(z)&(-e^{-1}\leq z<0),\\
W(z)&(z\geq0)
\end{cases}
\end{align*}
is real analytic on $[-e^{-1},\infty)$.
The Taylor series of $W(z)$ about 0 is 
\begin{align*}
\sum_{n=0}^\infty(-1)^{n+1}\frac{n^{n-1}}{n!}z^n
\end{align*}
with the radius $e^{-1}$ of convergence.

\begin{prp}\label{prp:Exactlsol}
Let
\begin{align}
\varphi(x)
=
-\frac{1}{k}\left[W\left(-\Psi^0e^{kG}\right)+1\right]
\label{eq:SolExact}
\end{align}
using the Lambert W function, where $\Psi^0$ is the initial value of the potential $\Psi$.
Then $\varphi$ solves the initial value problem of the exact differential equation \eqref{eq:Exact}, which is equivalent to \eqref{eq:ExactAbel} with imposing \eqref{eq:f1f2f3ODE}.
\end{prp}

\begin{proof}
First note that, given initial value $x_0$ of $x$, the solution curve, or the equipotential curve, of the exact differential equation \eqref{eq:Exact} is given by
\begin{align}
\Psi(x,\varphi)
&=
(1+k\varphi)\varpi(x,\varphi)
=
\Psi(x_0,\varphi(x_0))
=\Psi^0.
\label{eq:SolCurveExact}
\end{align}
Since $\varpi(x,\varphi)=e^{-(1+k\varphi)-kG}$,  we have
\begin{align}
-(1+k\varphi)e^{-\left(1+k\varphi\right)}
&=
-\Psi^0e^{kG}.
\label{eq:Psi0}
\end{align}
The solution curve is the Lambert curve on the $(z,w)$-plane by imposing
\begin{align*}
z=-\Psi^0e^{kG}
\quad
\mbox{and}
\quad
w=-(1+k\varphi).
\end{align*}
It immediately follows
\begin{align}
-(1+k\varphi)
&=
W\left(-\Psi^0e^{kG}\right)
\label{eq:LambertWPsi}
\end{align}
from \eqref{eq:Psi0}.
Thus the solution to the exact differential equation \eqref{eq:Exact} is given by \eqref{eq:SolExact}.
\end{proof}

\begin{crl}\label{crl:Abelsol}
Let
\begin{align}
\phi(x)
=
\frac{F}{\varphi}
=
-\frac{kF}{W\left(-\Psi^0e^{kG}\right)+1}.
\label{eq:SolAbel}
\end{align}
Then $\phi$ solves the initial value problem of Abel's equation \eqref{eq:AbelODE1k} of the first kind with imposing \eqref{eq:f0ODE} and \eqref{eq:f1f2f3ODE}.
\qed
\end{crl}

Remark that the real-valued Lambert W function $W(z)$ is defined only for $z\in[-e^{-1},\infty)\subset\R$.
Moreover, $W(z)$ is double-valued and has the branch point only at $z=-e^{-1}$.
Given initial value $x_0$ of $x$, the initial value $\Psi^0$  of the potential $\Psi(x,\varphi)$ at $(x_0,\varphi(x_0))$ is uniquely determined.
Therefore, the solutions \eqref{eq:SolExact} and \eqref{eq:SolAbel} potentially possess the branch point only at $x$ that satisfies
\begin{align}
\Psi^0
&=
e^{-kG(x)-1}.
\label{eq:branchExact}
\end{align}
The branches of $W(z)$ in \eqref{eq:SolExact} and \eqref{eq:SolAbel} are uniquely determined as follows.

(1)\
If $\Psi^0>0$ then the solution curve $\varphi=\varphi(x)$ is restricted to the region
\begin{align*}
\mathcal{S}:=\left\{(x,\varphi)\ |\ k\varphi>-1\right\}
\end{align*}
on the $(x,\varphi)$-plane, since $\varpi$ is always positive (see \eqref{eq:SolCurveExact}).
In this case, $W$ takes negative value (see \eqref{eq:LambertWPsi}), therefore, we have $\Psi^0\leq e^{-kG-1}$, since $W(z)$ is defined only on $z\geq-e^{-1}$.
Thus $\Psi^0$ obeys
\begin{align*}
0<
\Psi^0
\leq e^{-kG-1}.
\end{align*}

Moreover, there exists exactly a branch point $x$ when $\Psi^0>0$ (see \eqref{eq:branchExact}), thereby $\varphi(x)$ is double-valued on $\mathcal{S}$.
Let
\begin{align*}
\mathcal{S}_-:=\left\{(x,\varphi)\ |\ -1<k\varphi\leq0\right\}\subset\mathcal{S}.
\end{align*}
Since $\varphi(x)$ is given by \eqref{eq:SolExact}, we choose the branch $W_0$ such that $W_0\geq-1$ for $(x,\varphi)\in\mathcal{S}_-$ to give the solution \eqref{eq:SolExact}.
Whereas, we choose another branch $W_{-1}$ such that $W_{-1}<-1$ for $(x,\varphi)\in\mathcal{S}_+$ to give \eqref{eq:SolExact}, where we let
\begin{align*}
\mathcal{S}_+:=\mathcal{S}\backslash\mathcal{S}_-=\left\{(x,\varphi)\ |\ k\varphi>0\right\}.
\end{align*}

(2)\
If $\Psi^0<0$ then the solution curve is contained in
\begin{align*}
\mathcal{T}:=\left\{(x,\varphi)\ |\ k\varphi<-1\right\}.
\end{align*}
For such an initial value $\Psi^0$ of $\Psi$, $W$ in \eqref{eq:SolExact} necessarily takes positive value.
This is consistent with the positivity of the independent variable $-\Psi^0e^{kG(x)}$ of $W$, because $W$ is positive single-valued on the positive real axis.
Hence,  for $(x,\varphi)\in\mathcal{T}$, the solution \eqref{eq:SolExact} is uniquely given by the single-valued $W$ defined on $[0,\infty)$.

(3)\
Finally, if $\Psi^0=0$ then $(x,\varphi)$ is on the boundary between $\mathcal{S}$ and $\mathcal{T}$.
Hence, $k\varphi(x)\equiv-1$, thereby $\Psi\equiv0$.
This gives a constant solution to \eqref{eq:ExactAbel}, which is a special solution contained in \eqref{eq:SolExact}. 

Above discussion is summarized into the following:

\begin{prp}\label{prop:potential}
Assume \eqref{eq:f0ODE} and \eqref{eq:f1f2f3ODE} to be satisfied in order that Abel's equation \eqref{eq:ExactAbel} of the first kind reduces to the exact differential equation \eqref{eq:Exact}.
Let $\Psi^0$ be the initial value of the potential $\Psi$ to \eqref{eq:Exact}.
\begin{itemize}
\item[(1)] If $\Psi^0>0$, the Lambert W function $W$ in the solution \eqref{eq:SolExact} to \eqref{eq:ExactAbel} is either 
\begin{itemize}
\item[(a)] the branch $W_0$ for $(x,\varphi)\in\mathcal{S}_-$ or
\item[(b)] the branch $W_{-1}$ for $(x,\varphi)\in\mathcal{S}_+$, 
\end{itemize}
both of which are defined on $[-e^{-1},0)$.

\item[(2)] If $\Psi^0\leq0$ then $W$ in \eqref{eq:SolExact} is the single-valued W function defined on $[0,\infty)$.
\qed
\end{itemize}
\end{prp}

\subsection{The SIRv model and its general solution}\label{subsec:SIRv}

Now, we relate Abel's equation (\ref{eq:AbelODE1k}) of the first kind to the SIRv model (\ref{eq:VL2ffeven2dim1}---\ref{eq:VL2ffeven2dim2}), following the approach of \cite{HLM14}.
Introduce a new variable $t=t(x)$ as an antiderivative of $\phi(x)$:
\begin{align*}
t&=
\int\phi(x)dx.
\end{align*}
Then we find
\begin{align*}
\frac{dx}{dt}
&=
\frac{1}{\DIS\frac{dt}{dx}}
=
\frac{1}{\phi}
\qquad
\mbox{and}
\\
\frac{d^2x}{dt^2}
&=
\frac{d}{dt}\frac{1}{\phi}
=
-\frac{1}{\phi^2}\frac{dx}{dt}\frac{d\phi}{dx}
=
-\frac{1}{\phi^3}\frac{d\phi}{dx}.
\end{align*}
It follows that we have
\begin{align*}
\phi
&=
\frac{1}{x^\prime}
\qquad
\mbox{and}
\qquad
\frac{d\phi}{dx}
=
-\frac{x^{\prime\prime}}{\left(x^\prime\right)^3},
\end{align*}
where we denote the derivative with respect to $t$ by ${}^\prime$ for simplicity.

Then Abel's equation \eqref{eq:AbelODE1k} of the first kind reduces to the second order ODE
\begin{align*}
x^{\prime\prime}
+
f_0\left(x^\prime\right)^3
+
f_1\left(x^\prime\right)^2
+
f_2x^\prime
+
f_3
&=
0.
\end{align*}
Dividing by $x$, we obtain
\begin{align}
\left(\log x\right)^{\prime\prime}
+
f_2\left(\log x\right)^\prime
+
\frac{f_0}{x}\left(x^\prime\right)^3
+
\frac{xf_1+1}{x^2}\left(x^\prime\right)^2
+
\frac{f_3}{x}
&=
0.
\label{eq:AbelODEvlog}
\end{align}

Let us consider the following system of first-order ODEs
\begin{align}
x^\prime
&=
-xy+ax+b,
\label{eq:SODE1}
\\
y^\prime
&=
xy+cy+d,
\label{eq:SODE2}
\end{align}
where $a,b,c,d\in\R$.
From the first equation \eqref{eq:SODE1}, we have
\begin{align*}
y
&=
a+\frac{b}{x}-\left(\log x\right)^\prime.
\end{align*}
By substituting this into the second equation \eqref{eq:SODE2}, we get
\begin{align*}
\left(\log x\right)^{\prime\prime}-c\left(\log x\right)^\prime-x\left(\log x\right)^\prime
+
\frac{b}{x^2}+a\left(x+c\right)+\frac{b}{x}\left(x+c\right)+d
&=
0.
\end{align*}
Comparing this with \eqref{eq:AbelODEvlog}, we see that, if
\begin{align*}
f_0
&=0,
\quad
f_1
=
-\frac{1}{x},
\quad
f_2
=
-x-c,
\quad
\mbox{and}
\quad
\\
f_3
&=
a(x+c)x+b\left(x+c+\frac{1}{x}\right)+dx,
\end{align*}
the system (\ref{eq:SODE1}--\ref{eq:SODE2}) of ODEs reduces to Abel's equation \eqref{eq:AbelODE1k} of the first kind, which satisfies \eqref{eq:f0ODE}.
Moreover, if $b=d=0$, the condition \eqref{eq:f1f2f3ODE} is also satisfied:
\begin{align*}
\frac{d}{dx}\log\frac{f_2}{f_3}
&=
\frac{d}{dx}\log\frac{1}{-a x}
=
-\frac{1}{x}
=
f_1,
\end{align*}
thereby (\ref{eq:SODE1}--\ref{eq:SODE2}) reduces to the exact differential equation \eqref{eq:Exact}.

Thus, we obtain the following proposition that claims a relation between Abel's equation of the first kind and the SIRv model, thereby the integrability of the SIRv model via the exact differential equation.

\begin{prp}
Let
\begin{align}
f_0
=
0,
\quad
f_1
=
-\frac{1}{x},
\quad
f_2
=
\gamma-x,
\quad
\mbox{and}
\quad
f_3
=
\nu\left(\gamma-x\right)x,
\label{eq:Abelsfs}
\end{align}
where $\gamma$ and $\nu$ are real numbers, in particular, $\nu\neq0$, since we assume $f_3\not\equiv0$.
Then, \eqref{eq:AbelODE1k} reduces to
\begin{align}
\frac{d\phi}{dx}
&=
-\frac{1}{x}\phi
+
(\gamma-x)\phi^2
+
\nu\left(\gamma-x\right)x\phi^3.
\label{eq:AbelODE1kSIR}
\end{align}
The Abel equation \eqref{eq:AbelODE1kSIR} of the first kind is equivalent to
\begin{align}
x^\prime
&=
-xy-\nu x,
\label{eq:SIRv1}
\\
y^\prime
&=
xy-\gamma y,
\label{eq:SIRv2}
\end{align} 
which is (\ref{eq:SODE1}--\ref{eq:SODE2}) with imposing
\begin{align*}
a=-\nu,
\quad
b=0,
\quad
c=-\gamma,
\quad
\mbox{and}
\quad
d=0.
\end{align*}
Moreover, (\ref{eq:SIRv1}--\ref{eq:SIRv2}) attributes to the exact differential equation \eqref{eq:Exact}.
\qed
\end{prp}

Suppose $x$ and $y$ respectively stand for the populations of the susceptible and the infected of an infectious disease.
Also let $\gamma$ and $\nu$ respectively represent the transmission rate and the vaccination rate.
Then the system (\ref{eq:SIRv1}--\ref{eq:SIRv2}) of ODEs governs the spread of the infectious diseases under the influence of vaccination, and is called the SIRv model \cite{Hethcote74,Nobe23}.
Thus, when we call (\ref{eq:SIRv1}--\ref{eq:SIRv2}) the SIRv model, we assume $\gamma$ and $\nu$ to be positive, since they respectively stand for the rates of transmission and vaccination.

The discussion in the previous section leads to the following proposition concerning the integrability of the SIRv model on the Poisson mmanifold.

\begin{prp}\label{SIRvVL}
Let
\begin{align*}
x_0=\alpha=\gamma,
\quad
x_1=x,
\quad
x_2=\beta=-\nu,
\quad
\mbox{and}
\quad
y_1=y.
\end{align*}
Then the SIRv model (\ref{eq:SIRv1}--\ref{eq:SIRv2}) is the completely integrable Volterra lattice (\ref{eq:VL2ffeven2dim1}--\ref{eq:VL2ffeven2dim2}) with the constant boundary (\III) on the Poisson manifold $(V=\R^2,\{\ ,\ \}_2)$ possessing the Hamiltonian
\begin{align}
H_{01}(x,y)
&=
x+y-\gamma\log x+\nu\log y.
\label{eq:HamiltonianSIRv}
\end{align} 
The SIRv model (\ref{eq:SIRv1}--\ref{eq:SIRv2}) is also the Hamiltonian flow arising from the symplectic structure 
\begin{align*}
\Omega=\frac{dx\wedge dy}{xy}
\end{align*}
on the symplectic manifold $(V,\Omega)$.
\qed
\end{prp}

\begin{rem}
Suppose $\nu=0$ against our assumption $f_3=\nu\left(\gamma-x\right)x\not\equiv0$.
Then (\ref{eq:SIRv1}--\ref{eq:SIRv2}) reduces to the original SIR mode.
Unfortunately, the discussion above, based on Abel's equation of the first kind, is not valid for the SIR model.
The SIR model can be investigated using the Riccati equation instead of the Abel equation \cite{Nobe23}.
\end{rem}

The SIRv model attributes to Abel's equation \eqref{eq:AbelODE1kSIR} of the first kind employing the coefficients \eqref{eq:Abelsfs}, which reduces to the exact differential equation \eqref{eq:Exact}.
Hence the SIRv model (\ref{eq:SIRv1}--\ref{eq:SIRv2}) and the Volterra lattice (\ref{eq:VL2ffeven2dim1}--\ref{eq:VL2ffeven2dim2}) are exactly solved as follows.

We have
\begin{align*}
F
&=
\exp\int f_1dx
=
\frac{C_F}{x},
\\
G
&=
\int f_2Fdx
=
C_F\left(\gamma\log x-x\right)+C_G,
\end{align*}
where $C_F$ and $C_G$ are integration constants.
Since $e^{C_G}$ acts on $\varpi$, or $\Psi$, as a constant multiplication, it does not affects the exact differential equation $d\Psi=0$.
Hence, we can assume $C_G=0$ without loss of generality.
Moreover, since we define the constant $k$ to be the integration constant
\begin{align*}
k
&=
\frac{f_2}{f_3F}
=
\frac{1}{\nu C_F},
\end{align*}
we may assume $C_F=1$ by putting $k=1/\nu$.
Thus we obtain
\begin{align*}
F
=
\frac{1}{x},
\qquad
G
=
\gamma\log x-x,
\qquad
\mbox{and}
\qquad
k
=
\frac{1}{\nu}.
\end{align*}

Then the equation \eqref{eq:ExactAbel} concerning $\varphi$ reduces to
\begin{align}
\varphi\frac{d\varphi}{dx}+\frac{\gamma-x}{x}\varphi+\nu\frac{\gamma-x}{x}=0.
\label{eq:ExactAbelSIR}
\end{align}
This attributes to the exact differential equation
\begin{align*}
d\Psi(x,\varphi)
&=
d\left[\left(1+\frac{1}{\nu}\varphi\right)\varpi(x,\varphi)\right]
=
0
\end{align*}
with employing the integrating factor given by
\begin{align*}
\varpi(x,\varphi)
&=
\exp\left[-\left(1+\frac{1}{\nu}\varphi\right)-\frac{1}{\nu}\left(\gamma\log x-x\right)\right]
=
x^{-\frac{\gamma}{\nu}}e^{\frac{x}{\nu}-\frac{\varphi}{\nu}-1}.
\end{align*}
The potential is
\begin{align}
\Psi(x,\varphi)
&=
\left(1+\frac{1}{\nu}\varphi\right)\varpi(x,\varphi)
=
-\frac{1}{\nu}\exp\left(\frac{H_{01}}{\nu}\right),
\label{eq:potensialSIRv}
\end{align}
where $H_{01}$ is the Hamiltonian \eqref{eq:HamiltonianSIRv}.
Thus the potential $\Psi$ is, of course, conserved through the evolution.
In addition, the exact differential equation $d\Psi=0$ is equivalent to Hamilton's canonical equation of motion,
\begin{align*}
dH_{01}(\p,\x)
&=
\frac{\partial H_{01}}{\partial\p} d\p+\frac{\partial H_{01}}{\partial\x} d\x
=
0,
\end{align*}
on the symplectic manifold $(V=\R^2,\Omega)$.

Remark that, for any initial values $x_0,y_0\in\R$, the initial value $\Psi^0$ of the potential is negative when $\nu>0$.
Hence, if $\nu>0$, by virtue of Proposition \ref{prop:potential}, the solution \eqref{eq:SolExact} to \eqref{eq:ExactAbelSIR} is given by the single-valued Lambert W function $W$ on $[0,\infty)$.
Explicitly, we have
\begin{align*}
\varphi(x)
&=
-\nu\left[W\left(\frac{1}{\nu}x^{\frac{\gamma}{\nu}}e^{\frac{1}{\nu}(H_{01}^0-x)}\right)+1\right],
\end{align*}
where $H_{01}^0:=H_{01}(x_0,y_0)$.
Similarly, the solution \eqref{eq:SolAbel} to Abel's equation \eqref{eq:AbelODE1kSIR} of the first kind is also provided by the single-valued $W$:
\begin{align*}
\phi(x)
=
-\frac{1}{\DIS\nu x\left[W\left(\frac{1}{\nu}x^{\frac{\gamma}{\nu}}e^{\frac{1}{\nu}(H_{01}^0-x)}\right)+1\right]}.
\end{align*}

Since the independent variable $t$ is the antiderivative of $\phi$, the solution to the SIRv model (\ref{eq:SIRv1}--\ref{eq:SIRv2}) is implicitly given by
\begin{align}
t
&=
-\int\frac{dx}{\DIS\nu x\left[W\left(\frac{1}{\nu}x^{\frac{\gamma}{\nu}}e^{\frac{1}{\nu}(H_{01}^0-x)}\right)+1\right]}.
\label{eq:tx1sol}
\end{align}
Moreover, since we have
\begin{align*}
y=-\nu-\frac{x^\prime}{x}
\end{align*}
form \eqref{eq:SIRv1}, we obtain
\begin{align}
y
&=
\nu W\left(\frac{1}{\nu}x^{\frac{\gamma}{\nu}}e^{\frac{1}{\nu}(H_{01}^0-x)}\right),
\label{eq:solSIRvy}
\end{align}
if $x$ satisfies \eqref{eq:tx1sol}.

We summarize the discussion above into the following:

\begin{thm}\label{thm:SolSIRvl}
The SIRv model (\ref{eq:SIRv1}--\ref{eq:SIRv2}) reduces to the exact differential equation $d\Psi(x,\varphi)=0$ for the potential \eqref{eq:potensialSIRv} by introducing $\varphi=(\log x)^\prime$.
The solution to the initial value problem of the SIRv model is implicitly provided by \eqref{eq:tx1sol} and \eqref{eq:solSIRvy}, where $H_{01}^0$ is the initial value of the conserved quantity.
Remark that $W$ is the single-valued Lambert W function defined on $[0,\infty)$, since we assume $\nu>0$.
\qed
\end{thm}

For the Volterra lattice (\ref{eq:VL2ffeven2dim1}--\ref{eq:VL2ffeven2dim2}), we have 
\begin{align*}
k=-\frac{1}{\beta},
\quad
\varphi&=\frac{\dot x_1}{x_1}=\beta-y_1,
\quad
\mbox{and}
\quad
\Psi=\frac{1}{\beta}\exp\left(-\frac{H_{01}}{\beta}\right).
\end{align*}
Remark that we assume $\beta\neq0$.
If and only if $\beta<0$, $\Psi^0<0$.
Hence, according to Proposition \ref{prop:potential}, $(x_1,\varphi(x_1))=(x_1,\beta-y_1)$ is contained in $\mathcal{T}$, thereby $(x_1,y_1)$ is in
\begin{align*}
\widetilde{\mathcal{T}}
&:=
\left\{
(x_1,y_1)\ |\ 
y_1>\beta
\right\}.
\end{align*}
The solution $(x_1,y_1)$ to the Volterra lattice (\ref{eq:VL2ffeven2dim1}--\ref{eq:VL2ffeven2dim2}) is given by the single-valued Lambert W function.

Whereas, if and only if $\beta>0$, $\Psi^0>0$.
Hence, according to Proposition \ref{prop:potential}, $(x_1,\varphi(x_1))=(x_1,\beta-y_1)$ is contained in $\mathcal{S}=\mathcal{S}_-\cup\mathcal{S}_+$, thereby $(x_1,y_1)$ is in
\begin{align*}
\widetilde{\mathcal{S}}
&=
\widetilde{\mathcal{S}}_-\cup\widetilde{\mathcal{S}}_+,
\\
\widetilde{\mathcal{S}}_-
&:=
\left\{
(x_1,y_1)\ |\ 
0<y_1\leq\beta
\right\},
\\
\widetilde{\mathcal{S}}_+
&:=
\left\{
(x_1,y_1)\ |\ 
y_1>\beta
\right\}.
\end{align*}
The solution $(x_1,y_1)$ to the Volterra lattice (\ref{eq:VL2ffeven2dim1}--\ref{eq:VL2ffeven2dim2}) is given by the branch $W_0$ of the Lambert W function if $(x_1,y_1)\in\widetilde{\mathcal{S}}_-$, while by the branch $W_{-1}$ if $(x_1,y_1)\in\widetilde{\mathcal{S}}_+$.

Thus, we obtain:
\begin{thm}\label{thm:SolVL}
The general solution to the initial value problem of the Volterra lattice (\ref{eq:VL2ffeven2dim1}--\ref{eq:VL2ffeven2dim2}) is implicitly provided by
\begin{align*}
t
&=
\int\frac{dx_1}{\DIS\beta x_1\left[W\left(-\frac{1}{\beta}x_1^{-\frac{\alpha}{\beta}}e^{\frac{1}{\beta}(x_1-H_{01}^0)}\right)+1\right]},
\\
y_1
&=
-\beta W\left(-\frac{1}{\beta}x_1^{-\frac{\alpha}{\beta}}e^{\frac{1}{\beta}(x_1-H_{01}^0)}\right),
\end{align*}
where $H_{01}^0$ is the initial value of the Hamiltonian $H_{01}$.
If $\beta<0$, the function $W$ is the single-valued Lambert W function, whereas, if $\beta>0$, $W$ is the branch $W_0$ of the Lambert W function when $(x_1,y_1)\in\widetilde{\mathcal{S}}_-$, and is the branch $W_{-1}$ when $(x_1,y_1)\in\widetilde{\mathcal{S}}_+$.
\qed
\end{thm}

\section{Concluding remarks}\label{sec:concl}
We investigate the complete integrability of the Volterra lattice with imposing the constant boundary condition (\III) in terms of the Poisson structure $\{\ ,\ \}_2$ and the symplectic structure $\Omega$ on the phase space $V=\R^M$.
We then find that the Volterra lattice has the symplectic structure $\Omega$ if $M=2N$ for any boundary values $\alpha,\beta$, and hence it achieves the complete integrability if $M=2$.
Such a Volterra lattice is nothing but the SIRv model, an integrable extension of the SIR epidemic model under the influence of vaccination.
Wheres, if $M=2N-1$ and $\alpha=\beta$, the Voltera lattice has the Poisson structure $\{\ ,\ \}_2$, and hence it also admits the complete integrability if $M=3$ and $\alpha=\beta$.
While such a Volterra lattice can also be seen as an integrable extension of the SIR model, its significance as an epidemic model has not yet been revealed, as far as the author knows.

Meanwhile, upon the introduction of an appropriate variable transformation, the SIRv model is transformed into Abel's equation of the first kind, which attributes to an exact differential equation.
The potential of the exact differential equation is, of course, the conserved quantity of the SIRv model, \textit{i.e.}, the Hamiltonian $H_{01}$ of the Volterra lattice on the symplectic manifold $(\R^2,\Omega)$.
Thus, the exact differential equation is equivalent to the Hamiltonian flow on $(\R^2,\Omega)$ with the Hamiltonian $H_{01}$. 
In addition, the invariant curve of the SIRv model, or the equipotential curve of the exact differential equation, is provided by the Lambert curve. 
Thus, we implicitly obtain the general solution to the initial value problem of the SIRv model, or the Volterra lattice on $(\R^2,\Omega)$, in terms of the Lambert W function.

The integrable discretization of completely integrable systems has been extensively studied for several decades. 
Regarding SIR epidemic models, there has been enthusiastic investigation into integrable discretization, resulting in the discovery of several discrete models that possess complete integrability (see \cite{WGCR03, SWRGC04, SIN18, TM22}).
In particular, an integrable discretization of the SIRv model, which preserves the same conserved quantity as the continuous model, has been achieved through a geometric construction utilizing its invariant curve. Although the process of geometric discretization was omitted in this article due to space limitations, it has been thoroughly documented in \cite{Nobe23}. Interested readers are encouraged to refer to this article for further insight into the geometric discretization of the SIRv model.



\begin{thebibliography}{99}
%
%
\bibitem{KM75-1}
Kac,~M. and van Moerbeke,~P.,
Some probabilistic aspects of scattering theory, 
in Arthurs,~A.M. (ed.), 
\textit{Functional integration and its applications}, 
Clarendon Press, (1975) 87--96.

\bibitem{KM75-2}
Kac,~M. and van Moerbeke,~P.,
On an explicitly soluble system of nonlinear differential equations related to certain Toda lattices, 
\textit{Adv. Math.}, 
\textbf{16} (1975) 160--169.

\bibitem{Moser75}
Moser,~J, 
Finitely many mass points on the line under the influence of an exponential potential--an integrable system, 
in Moser,~J. (ed.),
\textit{Dynamical systems, theory and applications}, 
Lecture Notes in Physics, \textbf{38}, 
Springer--Verlag, (1975) 467--497.

\bibitem{Bogoyavlenskij91} 
Bogoyavlenskij,~O.I., 
Algebraic constructions of integrable dynamical systems--extensions of the Volterra lattice,
\textit{Russian Math. Surveys},
\textbf{46} (1991), 1--64.

\bibitem{Bogoyavlenskij08} 
Bogoyavlenskij,~O.I., 
Integrable {L}otka-{V}olterra Systems,
\textit{Regul. Chaotic Dyn.},
\textbf{13} (2008), 543--556.

\bibitem{Suris03}
Suris,~Y.B.,
\textit{The Problem of Integrable Discretization: Hamiltonian Approach}, Progress in Mathematics \textbf{219},
Birkh\"auser, 2003.

\bibitem{Hethcote74}
Hethcote,~H.W.,
Asymptotic behavior and stability in epidemic models,
\textit{Mathematical Problems in Biology}, Lecture Notes in Biomathematics \textbf{2},
Springer--Verlag, (1974) 83--92.

\bibitem{KM27} 
Kermack,~W. and McKendrick,~A.,
A contribution to the mathematical theory of epidemics,
\textit{Proc. Roy. Soc. Lond A},
\textbf{115} (1927) 700--721.

\bibitem{HLM14}
Harko,~T., Lobo,~F. and Mak,~M.,
Exact analytical solutions of the {S}usceptible-{I}nfected-{R}ecovered ({SIR}) epidemic model and of the {SIR} model with equal death and birth rate,
\textit{Appl. Math. Comput.},
\textbf{236} (2014) 184--194.

\bibitem{MR16}
Mancas,~S. and Rosu,~H.,
Integrable {A}bel equations and {V}ein's equation,
\textit{Math. Meth. Appl. Sci.},
\textbf{39} (2016) 1376--1387.

\bibitem{CGHJK96}
Corless,~R.M., Gonnet,~G.H., Hare,~D.E.G., Jeffery,~D.J. and Knuth,~D.E.,
On the {L}ambert {W} functions,
\textit{Adv. Compt. Math.},
\textbf{5} (1996) 329--360.

\bibitem{Nobe23} 
Nobe,~A., 
Exact solutions to SIR epidemic models via integrable discretization,
\textit{arXiv: 2303.17198} (2023).

\bibitem{WGCR03}
Willox,~R., Grammaticos,~B., Carstea,~A.S. and Ramani,~A.,
Epidemic dynamics: discrete-time and cellular automaton models
\textit{Physica A},
\textbf{328} (2003) 13--22.

\bibitem{SWRGC04}
Satsuma,~J., Willox,~R., Ramani,~A., Grammaticos,~B. and Carstea,~A.S.,
Extending the {SIR} epidemic model,
\textit{Physica A},
\textbf{336} (2004) 369--375.

\bibitem{SIN18}
Sekiguchi,~M. Ishiwata,~E. and Nakata,~Y.,
Dynamics of an ultra-discrete SIR epidemic model with time delay,
\textit{Math. Biosci. Eng.},
\textbf{15} (2018) 653--666.

\bibitem{TM22}
Tanaka,~Y. and Maruno,~K.,
Integrable discretizations of the {SIR} model,
\textit{arXiv:2209.08549} (2022).

\end{thebibliography}
\end{document}